\documentclass[A4,12pt]{article}
\usepackage{amsmath,amssymb,amsfonts,amsthm,graphicx}
\usepackage{multirow}
\usepackage{diagbox}
\usepackage{color}
\usepackage{cite}
\usepackage{epsfig}
\usepackage{epstopdf}
\usepackage{footnote}
\usepackage{cite}
\usepackage{caption}
\usepackage{subcaption}         
\usepackage{graphicx}
\usepackage{enumerate}
\usepackage{times}
\usepackage{algorithm}
\usepackage{algpseudocode}
\usepackage{caption}

\usepackage[top=2cm, bottom=2cm, left=2cm, right=2cm]{geometry}

\newtheorem{theorem}{Theorem}[section]
\newtheorem{lemma}[theorem]{Lemma}

\theoremstyle{definition}
\newtheorem{definition}[theorem]{Definition}

\newtheorem{example}[theorem]{Example}

\theoremstyle{remark}

\title{An Explicit Construction of Optimal Dominating and $[1,2]-$Dominating Sets in Grid}
\author{P. Sharifani$^{1}$, M.R. Hooshmandasl$^{2}$, M. Alambardar Meybodi$^{3}$ \\
\footnotesize{$^{1,2,3}$Department of Computer Science, Yazd University, Yazd, Iran.}\\
\footnotesize{$^{1,2,3}$The Laboratory of Quantum Information Processing, Yazd University, Yazd, Iran.}  \\
\footnotesize{e-mail:$^1$ pouyeh.sharifani@gmail.com, $^2$ hooshmandasl@yazd.ac.ir, $^3$ 	m.alambardar@stu.yazd.ac.ir}}
\date{}

\begin{document}

	\maketitle
	
	\begin{abstract}		
		A dominating set in a graph $G$ is a subset of vertices $D$ such that
		every vertex in $V\setminus D$ is a neighbor of some vertex of $D$. The domination number of $G$  is the minimum size of a dominating set of $G$ and it is denoted by $\gamma(G)$. Also, a subset $D$ of a graph $G$ is a $[ 1 , 2 ] $-set if, each vertex $v \in V \setminus D$ is adjacent to either one or two vertices in $D$ and the minimum cardinality of $[ 1 , 2 ] $-dominating set of $G$, is denoted by $\gamma_{[1,2]}(G)$.
		 Chang's conjecture says that for every $16 \leq m \leq n, \gamma(G_{m,n})= \left \lfloor\frac{(n+2)(m+2)}{5}\right \rfloor-4$ and this conjecture has been proven by Goncalves et al. This paper presents an explicit constructing method to find an optimal dominating set for grid graph $G_{m,n}$ where $m,n\geq 16$ in $O(\text{size of answer})$. In addition, we will show that $\gamma(G_{m,n})=\gamma_{[1,2]}(G_{m,n})$ where $m,n\geq 16$ holds in response to an open question posed by Chellali et al.
		  \newline
		 		
	\end{abstract}

\section{Introduction}
The concept of domination and dominating set is a well-studied topic in graph theory and has
many extensions and applications.  A discussion of some of these can be found in \cite{haynes1998fundamentals, haynes1997domination}. The solution of
many variations of domination problems have shown to be NP-complete\cite{johnson1985np,masuyama1981computational,lichtenstein1982planar,attalah2013connected,chellali20131,chellali2014independent}. Also, many
algorithmic results have studied for these problems in different classes of graphs.
A subset $S$ of vertices is a dominating set if every vertex not in $S$ has at least one neighbor in $S$.
A dominating set with minimum cardinality is called an
optimal dominating set of a graph $G$; its cardinality is called
the domination number of $G$ and is denoted by $\gamma(G)$. Note
that although the domination number of a graph, $\gamma(G)$, is
unique, there may be different optimal dominating sets.
Grid graphs are a special class of graphs and the dominating set of them have many applications in robotics and sensor networks. 
Due to the special structure of grids, their domination number can be determined optimally, although this number was known recently by Goncalves et al., \cite{gonccalves2011domination}. They proved that for $m\times n$ grids, where $m,n \geq 16$, 
$$\gamma(G_{m,n})= \left \lfloor\frac{(n+2)(m+2)}{5}\right \rfloor-4.$$

 The idea behind their proof is using a dynamic programming method to store all dominating sets that occur in borders of a grid.
Various attempts have been made in recent years to find an algorithm for the optimal dominating set. 
Chang \cite{chang1994domination}, by using diagonalization and projection, constructed a dominating set for grids in polynomial-time, such that 
$$\gamma(G_{m,n})\leq \left \lceil \frac{(n+2)(m+2)}{5}\right\rceil.$$

The cardinality of dominating set constructed by Chang's method is at most $\gamma(G_{m,n}) + 5$, when $16 \leq m \leq n$.

In \cite{alanko2011computing}, Alanko et al, used brute-force computational
technique to find optimal dominating set in grids of size up
to $n = m = 29$.

Fata et al. \cite{fata2013distributed} presented a distributed algorithm for
finding near optimal dominating sets on grids. The size of the dominating set
provided by their algorithm is upper-bounded by
$\left \lceil \frac{(n+2)(m+2)}{5} \right \rceil$
for $m\times n$ grids and its difference from the optimal
domination number of the grid is upper-bounded by five.

P. Pisantechakool et al. \cite{pisantechakool2015new} improved upon the  distributed algorithm of Fata et al. and presented a new distributed algorithm that computes
a dominating set of size $\left \lceil\frac{(n+2)(m+2)}{5}\right \rceil-3$ on an $m\times n$ grid, $8\leq m , n$ and its difference from the optimal
domination number of the grid is upper-bounded by two.

 There are numerous intermediate results for minimal dominating set and $\gamma(G_{m,n})$ for small values of $n$ and $m$ by a dynamic programming algorithm \cite{hare1986algorithms,hare1995application,livingston1994constant,ma1990partition,singh1987parallel,vzerovnik1999deriving}.
 

The most of those algorithms are not efficient in practice when the values $n$ or $m$ be over 20. 

One of the interesting issues that can be expressed in different types of domination problems is that under which conditions the domination number of graph is equal to the domination number of that particular domination problem \cite{henning2013graphs, hansberg2007characterization,chellali2012trees}.

A dominating set like $D$ of graph $G(V,E)$, called $[1,2]$-dominating set, if each vertex $v \in V \setminus D$ is adjacent to at most two vertices in $D$. The concept of $[1, 2]$- dominating set, as a special case of $[\rho,\sigma]$-dominating set \cite{telle1994complexity}, is introduced by Chellali et al\cite{chellali20131}. They studied $[ 1, 2 ]$- dominating sets in graphs and posed a number of open problems. Some of those problems are solved in \cite{yang20141,goharshady20161}. One of the proposed questions is:\\

{\bf Question:} Is it true for grid graphs $G_{m,n}$ that $\gamma(G_{m,n})=\gamma_{[1,2]}(G_{m,n})$?\\

We show that the answer to this question is positive by constructing a $\gamma$-set for $G_{m,n}$ which is also a $\gamma_{[1,2]}$-set.

The main result of this paper is a construction  to find an optimal dominating set for grid  $G_{m,n}$ where $m,n\geq 16$.
The rest of this paper proceeds as follows. In Section 2 we describe some notations and definitions are needed. In section 3, we present our construction and the correctness argument for it. Also to ease of understanding we illustrate some examples. Finally, in section 4 we prove that $\gamma(G_{m,n})=\gamma_{[1,2]}(G_{m,n})$ for $m,n\geq 16$.

\section{Terminology}
In this section, we introduce some definitions and notations that will be needed in the sequel.  For all terminologies and notations are not defined here, we refer to \cite{bondy2008graph}.
Let $G=(V,E)$ be a simple graph, the neighborhood of a vertex $v \in V$ is the set of all vertices adjacent to $v$ and is denoted by $N(v)$, i.e. $N(v) = \left\{ u \in V \vert uv \in E \right\}$. The closed neighborhood of a vertex $v$ is defined $N[v] = N(v) \cup \{v \}$. A set $S$ is called a dominating set of $G$ if every vertex is either in $S$ or adjacent to a vertex in $S$. The size of the smallest dominating sets of a graph $G$ is denoted by $\gamma (G)$. Any such set is called a $\gamma$-set or minimal dominating set of $G$.

A set $S \subseteq V$ is called a $[1, 2]$-set of $G$ if for each $v \in V \setminus S$ we have $1 \leq \vert N(v) \cap S \vert \leq 2$, i.e. $v$ is adjacent to at least one but not more than two vertices in $S$. The size of the smallest $[1, 2]$-sets of $G$ is denoted by $\gamma_{[1, 2]} (G)$. Any such set is called a $\gamma_{[1, 2]}$-set of $G$. 
	We know that  for every graph $G$, $\gamma(G) \leq \gamma_{[1,2]}(G)\leq n$ \cite{chellali20131}. In some classes of graphs the domination number and $[1,2]$-domination number are equal.  This equality holds for cycles, caterpillars, claw free graphs, $P_4$-free graphs and nontrivial graph $G$ with $\Delta(G)\geq |V(G)|-3$, are proved in \cite{chellali20131}.
	
An $m\times n$ grid graph $G_{m,n} = (V;E)$ has vertex set $V = \{v_{i,j}:\,\, 1 \leq i \leq  m ,\, 1 \leq  j \leq n \}$ and edge set $E = \{(v_{i,j},v_{i,j'}):\,\, \vert j-j' \vert =1\} \cup \{(v_{i,j},v_{i',j}):\,\, \vert i-i' \vert =1\}$.
For ease of exposition, we will fix an orientation and labeling of the vertices, so
that vertex $v_{1,1}$ is the upper-left vertex and vertex $v_{m,n}$ is
the lower-right vertex of the grid.

We also require the following definitions.
\begin{definition}
 The boundary of grid $G_{m,n}$, denoted by $B(G)$, is the set of vertices like $v\in V$ such that $|N(v)|<4.$
\end{definition} 

\begin{definition}
A sub-grid of $G_{m,n}$ is induced graph by vertices $V = \{v_{i,j}:\,\, 2 \leq i \leq  m-1 ,\, 2 \leq  j \leq n-1 \}\setminus \{v_{2,2},v_{2,n-1},v_{m-1,2},v_{m-1,n-1}\}$ (see Figure \ref{subgrid}).
\end{definition}

\begin{figure}[h!]
	\centering
	\includegraphics[scale=0.7]{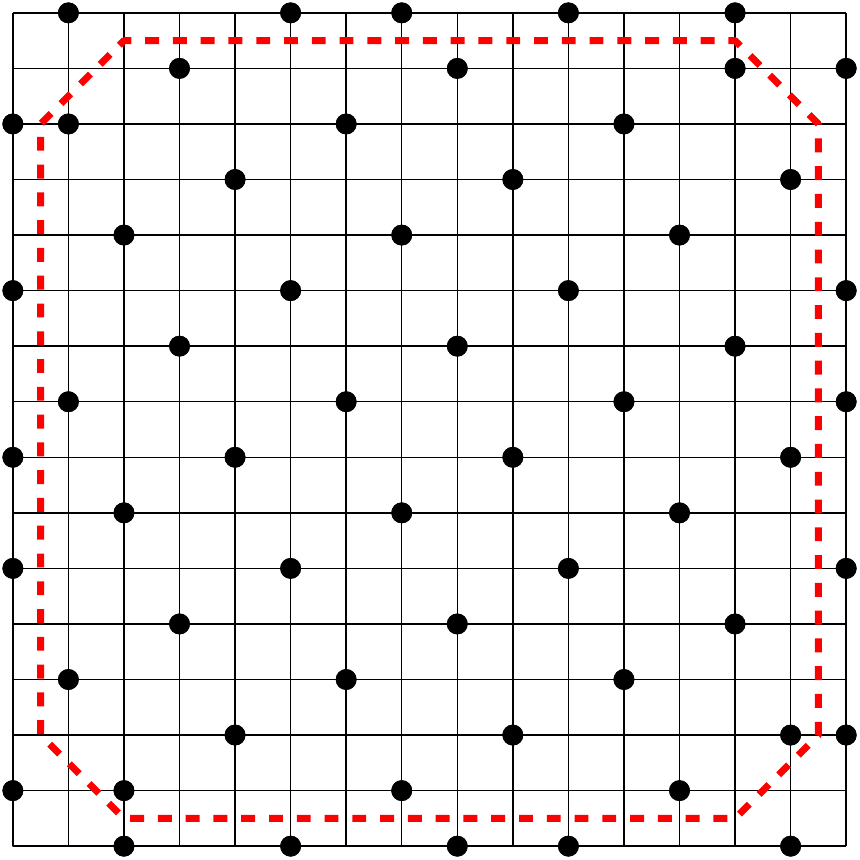}
	\caption{Example of dominating set of size 60 for grid $G_{16,16}$ and  its sub-grid 
	is highlighted by a red dashed line. The intersection of each row and column is a vertex of grid}\label{subgrid}
\end{figure}

\section{Construction of Dominating Set in Grid}\label{construct}

 The idea behind our method is choosing a proper pattern that dominates every vertex in sub-grid exactly once. Our purpose is selecting most of the vertices of dominating set from sub-grid as soon it is possible because every vertex in sub-grid dominates at most five vertices (its four neighbors and itself). These selected vertices are indicated by black disks.  Since, just by selecting vertices of sub-grid, the boundary vertices of the grid may be not dominated, then we have to add some vertices of the boundary, indicated by white squares, to the dominating set. To do so, we identify vertices of the optimal dominating set by two following steps:

{\bf Step 1:}  Identifying domination points (black disks) to dominate the vertices of sub-grid and some of the vertices of the boundary.

{\bf Step 2:}  Identifying domination points (white squares) to dominate the boundary vertices which are not dominated by black disks. 

In the step 1, according to the number of columns, we select the index of first appropriate column that a black disk  must be located in row $p$, denoted $a_p$. In other words, the vertex $v_{p,a_p}$ is the first position in row $p$ that a domination point is located.

For the first row, $a_p$ is defined as
\[
a_1=\left\{\begin{array}{lll}
2 &  & \text{if}\,\,  n\equiv_5 0,  \\
n \mod 5    & & \text{other wise,}
\end{array}\right.
\]
and for other rows, $a_p=(a_1+3(p-1)) \mod 5.$

Then, we construct the set $D_d$ as union of the following sets
\begin{align*}
D_F &=\{v_{1,5k+a_1}\,: \, 3\leq 5k+a_1\leq n-2 \;\; \text{for some}\,\, k\},\\
D_M &=\{v_{p,5k+a_p}\,: \, 2\leq p \leq n-1 \;\; \text{and}\;\; 1\leq 5k+a_1\leq n \;\; \text{for some}\,\, k\},\\
D_L  & =\{v_{m,5k+a_n}\,: \, 3\leq 5k+a_1\leq n-2 \;\; \text{for some}\,\, k\},
\end{align*}
where $D_F,D_M$ and $D_L$ are the sets of all black disks in first row, middle rows and last row respectively.

In the step 2, at first we define the set $A_{k}^{(i,j)}$  as
\[A_{k}^{(i,j)}=\{\, 5t+k\,|\, i\leq t\leq j\,\}.\]
Also, we set $S=\lfloor \frac{n}{5} \rfloor$ and $T=\lfloor \frac{m}{5} \rfloor$.

In order to cover vertices in the borders which are not dominated, the following vertices are added to the dominating set. These vertices are indicated by white squares.

{\bf First Row: } 
Selecting the white squares in first row is just depend on $n$ and is independent from  $m$.
 The set of all white squares of first row is defined by $D_s^{FR}$. This set is selected  as follow:
\[
D_s^{FR}=\left\{\begin{array}{lll}
\{v_{1,p}:\, p\in A_4^{(1,S-1)}\}\cup\{v_{1,3}\} & \text{if }&n\equiv_{_5 }0,  \\ \noalign{\medskip}
\{v_{1,p}:\, p\in A_3^{(1,S-2)}\}\cup\{v_{1,2},v_{1,n-2}\} & \text{if }&n\equiv_{_5 }1,  \\ \noalign{\medskip}
\{v_{1,p}:\, p\in A_4^{(1,S-2)}\}\cup\{v_{1,3},v_{1,n-2}\} & \text{if }&n\equiv_{_5 }2,  \\ \noalign{\medskip}
\{v_{1,p}:\, p\in A_0^{(1,S-1)}\}\cup\{v_{1,n-2}\} & \text{if }&n\equiv_{_5 }3,  \\ \noalign{\medskip}
\{v_{1,p}:\, p\in A_1^{(1,S-1)}\}\cup\{v_{1,n-2}\} & \text{if }&n\equiv_{_5 }4,  \\ \noalign{\medskip}
\end{array}\right.
\]

{\bf Left-Right Columns and Last Row: } 
Selecting the white squares in these borders, aside from the first row, depend on both $m$ and  $n$. The sets of all white squares in the first, last columns and last row are denoted by $D_s^{FC}$, $D_s^{LC}$ and $D_s^{LR}$, respectively. These points are selected  as follow:
			  
\begin{itemize}
\item[{\bf Case 1:}]$n \mod 5=0$
\begin{align*}
D_s^{FC} &=\left\{\begin{array}{lll} 
\{v_{p,1}:\, p\in A_2^{(0,T-2)}\}\cup\{v_{m-3,1}\}   &\text{if }  & m\equiv_{_5 }0, \\   \noalign{\medskip}
\{v_{p,1}:\, p\in A_2^{(0,T-1)}\}\hspace{2.7cm}  &\text{if }  & m\equiv_{_5 }1,2, \\   \noalign{\medskip}
\{v_{p,1}:\, p\in A_2^{(0,T)}\}  &\text{if } & m\equiv_{_5 }3. \\  \noalign{\medskip}
\{v_{p,1}:\, p\in A_2^{(0,T-1)}\} \cup\{v_{m-1,1}\}   &\text{if }  & m\equiv_{_5 }4, \\   \noalign{\medskip}
\end{array}\right.\\
D_s^{LC} &=\left\{\begin{array}{lll}
\{v_{p,n}:\, p\in A_4^{(1,T-1)}\}\cup \{v_{3,n}\} & \text{if }&m\equiv_{_5 }0, 3 , 4,  \\ \noalign{\medskip}
\{v_{p,n}:\, p\in A_4^{(1,T-2)}\}\cup \{v_{3,n}, v_{m-1,n}\} & \text{if }&m\equiv_{_5 }1 ,  \\  \noalign{\medskip}
\{v_{p,n}:\, p\in A_4^{(1,T-2)}\}\cup \{v_{3,n}, v_{m-2,n}\} & \text{if }&m\equiv_{_5 }2 ,  \\  \noalign{\medskip}
\end{array}\right.\\
D_s^{LR} &=\left\{\begin{array}{lll}
\{v_{m,p}:\, p\in A_2^{(0,S-2)}\}\cup \{v_{m,n-2}\} & \text{if }&m\equiv_{_5 }0,  \\ \noalign{\medskip}
\{v_{m,p}:\, p\in A_0^{(1,S-1)}\} & \text{if }&m\equiv_{_5 }1,  \\ \noalign{\medskip}
\{v_{m,p}:\, p\in A_1^{(1,S-1)}\} & \text{if }&m\equiv_{_5 }3,  \\ \noalign{\medskip}
\{v_{m,p}:\, p\in A_3^{(1,S-2)}\}\cup \{v_{m,2},v_{m,n-1}\} & \text{if }&m\equiv_{_5 }2,  \\ \noalign{\medskip}
\{v_{m,p}:\, p\in A_4^{(1,S)}\}\cup \{v_{m,3}\} & \text{if }&m\equiv_{_5 }4,  \\ \noalign{\medskip}
\end{array}\right.\\
\end{align*}

\item[{\bf Case 2:}]$n \mod 5=1$

\begin{align*}
D_s^{FC} &=\left\{\begin{array}{lll} 
\{v_{p,1}:\, p\in A_4^{(1,T-1)}\}\cup \{v_{3,1}\} &  \text{if }& m\equiv_{_5 }0,3,4 \\  \noalign{\medskip}
\{v_{p,1}:\, p\in A_4^{(1,T-2)}\}\cup \{v_{3,1},v_{m-1,1}\} &\text{if } &  m\equiv_{_5 }1, \\  \noalign{\medskip}
\{v_{p,1}:\, p\in A_4^{(1,T-2)}\}\cup \{v_{3,1},v_{m-2,1}\} & \text{if } &  m\equiv_{_5 }2,\\  \noalign{\medskip} 
\end{array}\right.\\
D_s^{LC} &=\left\{\begin{array}{lll}
\{v_{p,n}:\, p\in A_3^{(1,T-2)}\}\cup \{v_{2,n},v_{m-1,n}\} & \text{if }& m\equiv_{_5 }0,  \\ \noalign{\medskip}
\{v_{p,n}:\, p\in A_3^{(1,T-2)}\}\cup \{v_{2,n}\} & \text{if }& m\equiv_{_5 }2,3,  \\ \noalign{\medskip}
\{v_{p,n}:\, p\in A_3^{(1,T-2)}\}\cup \{v_{2,n},v_{m-2,n}\} & \text{if }& m\equiv_{_5 }1,  \\ \noalign{\medskip}
\{v_{p,n}:\, p\in A_3^{(1,T-1)}\}\cup \{v_{2,n},v_{m-2,n}\} & \text{if }& m\equiv_{_5 }4,  \\ \noalign{\medskip}
\end{array}\right.\\
D_s^{LR} &=\left\{\begin{array}{lll}
\{v_{m,p}:\, p\in A_1^{(1,S-1)}\} & \text{if } & m\equiv_{_5 }0,  \\ \noalign{\medskip}
\{v_{m,p}:\, p\in A_4^{(1,S-1)}\}\cup \{v_{m,3},v_{m,n-1}\} & \text{if }& m\equiv_{_5 }1,  \\ \noalign{\medskip}
\{v_{m,p}:\, p\in A_2^{(0,S-1)}\} & \text{if }& m\equiv_{_5 }2,  \\ \noalign{\medskip}
\{v_{m,p}:\, p\in A_0^{(1,S)}\} & \text{if }& m\equiv_{_5 }3,  \\ \noalign{\medskip}
\{v_{m,p}:\, p\in A_3^{(1,S-2)}\}\cup \{v_{m,2},v_{m,n-2}\} & \text{if }& m\equiv_{_5 }4,  \\ \noalign{\medskip}
\end{array}\right.
\end{align*}

\item[{\bf Case 3:}]$n \mod 5=2$

\begin{align*}
D_s^{FC} &=\left\{\begin{array}{lll} 
\{v_{p,1}:\, p\in A_2^{(0,T-2)}\}\cup \{v_{m-2,1}\} \hspace{9mm} &  \text{if } & m\equiv_{_5 }0, \\  \noalign{\medskip}
\{v_{p,1}:\, p\in A_2^{(0,T-1)}\}  &  \text{if} & m\equiv_{_5 }1,2,3, \\  \noalign{\medskip}
\{v_{p,1}:\, p\in A_2^{(0,T-1)}\} \cup \{v_{m-1,1}\}  & \text{if } & m\equiv_{_5 }4, \\  \noalign{\medskip} 
\end{array}\right.\\
D_s^{LC} &=\left\{\begin{array}{lll}
\{v_{p,n}:\, p\in A_3^{(1,T-2)}\}\cup \{v_{2,n}v_{m-1,n}\} \hspace{3mm} & \text{if} & m\equiv_{_5 }0,  \\ \noalign{\medskip}
\{v_{p,n}:\, p\in A_3^{(1,T-2)}\}\cup \{v_{2,n}v_{m-2,n}\} & \text{if} & m\equiv_{_5 }1,  \\ \noalign{\medskip}
\{v_{p,n}:\, p\in A_3^{(1,T-1)}\}\cup \{v_{2,n}\} & \text{if } & m\equiv_{_5 }2,3,  \\ \noalign{\medskip}
\{v_{p,n}:\, p\in A_3^{(1,T)}\}\cup \{v_{2,n}\} & \text{if } & m\equiv_{_5 }4,  \\ \noalign{\medskip}
\end{array}\right.\\
D_s^{LR} &=\left\{\begin{array}{lll}
	\{v_{m,p}:\, p\in A_2^{(0,S-1)}\} & \text{if } & m\equiv_{_5 }0,  \\ \noalign{\medskip}
\{v_{m,p}:\, p\in A_0^{(1,S-1 )}\}\cup \{v_{m,n-1}\} & \text{if } & m\equiv_{_5 }1,  \\ \noalign{\medskip}
\{v_{m,p}:\, p\in A_3^{(1,S-1)}\}\cup \{v_{m,2}\} & \text{if } & m\equiv_{_5 }2,  \\ \noalign{\medskip}
\{v_{m,p}:\, p\in A_1^{(1,S)}\} & \text{if } & m\equiv_{_5 }3,  \\ \noalign{\medskip}
\{v_{m,p}:\, p\in A_4^{(1,S-2)}\}\cup \{v_{m,3},v_{m,n-2}\} & \text{if } & m\equiv_{_5 } 4,  \\ \noalign{\medskip}
\end{array}\right.
\end{align*} 

\item[{\bf Case 4:}]$n \mod 5=3$

\begin{align*}
D_s^{FC} &=\left\{\begin{array}{lll} 
\{v_{p,1}:\, p\in A_0^{(1,T-1)}\}\hspace{2.8cm} & \text{if } & m\equiv_{_5 }0, \\  \noalign{\medskip}
\{v_{p,1}:\, p\in A_0^{(1,T)}\} &\text{if } & m\equiv_{_5 }1,4, \\  \noalign{\medskip}
\{v_{p,1}:\, p\in A_0^{(1,T-1)}\}\cup \{v_{m-1,1}\} &  \text{if } & m\equiv_{_5 }2,\\  \noalign{\medskip}
\{v_{p,1}:\, p\in A_0^{(1,T-1)}\} \cup \{v_{m-2,1}\} & \text{if } & m\equiv_{_5 }3, \\  \noalign{\medskip} 
\end{array}\right.\\
D_s^{LC} &=\left\{\begin{array}{lll}
\{v_{p,n}:\, p\in A_3^{(1,T-2)}\}\cup \{v_{2,n},v_{m-1,n}\} & \text{if } & m\equiv_{_5 }0,  \\ \noalign{\medskip}
\{v_{p,n}:\, p\in A_3^{(1,T-2)}\}\cup \{v_{2,n},v_{m-2,n}\} & \text{if } & m\equiv_{_5 } 1,  \\ \noalign{\medskip}
\{v_{p,n}:\, p\in A_3^{(1,T-1)}\}\cup \{v_{2,n}\} & \text{if } & m\equiv_{_5 }2,3,  \\ \noalign{\medskip}
	\{v_{p,n}:\, p\in A_3^{(1,T)}\}\cup \{v_{2,n}\} & \text{if } & m\equiv_{_5 }4,  \\ \noalign{\medskip}
\end{array}\right.\\
D_s^{LR} &=\left\{\begin{array}{lll}
	\{v_{m,p}:\, p\in A_3^{(1,S-1)}\} \cup \{v_{m,2}\}  \hspace{1.2cm} & \text{if } & m\equiv_{_5 }0,  \\ \noalign{\medskip}
\{v_{m,p}:\, p\in A_1^{(1,S-1)}\}\cup \{v_{m,n-1}\} & \text{if } & m\equiv_{_5 }1,  \\ \noalign{\medskip}
\{v_{m,p}:\, p\in A_4^{(1,S-1)}\}\cup \{v_{m,3}\} & \text{if } & m\equiv_{_5 }2,  \\ \noalign{\medskip}
\{v_{m,p}:\, p\in A_2^{(0,S)}\} & \text{if } & m\equiv_{_5 }3,  \\ \noalign{\medskip}
\{v_{m,p}:\, p\in A_0^{(1,S-1)}\}\cup \{v_{m,n-2}\} & \text{if } & m\equiv_{_5 }4,  \\ \noalign{\medskip}
\end{array}\right.
\end{align*}

\item[{\bf Case 5:}]$n \mod 5=4$

\begin{align*}
D_s^{FC} &=\left\{\begin{array}{lll} 
\{v_{p,1}:\, p\in A_3^{(1,T-2)}\} \cup \{v_{2,1},v_{m-1,1}\} \hspace{1mm} &  \text{if } & m\equiv_{_5 }0, \\  \noalign{\medskip}
\{v_{p,1}:\, p\in A_3^{(1,T-2)}\} \cup \{v_{2,1},v_{m-2,1}\} & \text{if } & m\equiv_{_5 }1, \\  \noalign{\medskip}
\{v_{p,1}:\, p\in A_3^{(1,T-1)}\} \cup \{v_{2,1}\} & \text{if } & m\equiv_{_5 }2,3,4,\\  \noalign{\medskip}
\end{array}\right.\\
D_s^{LC} &=\left\{\begin{array}{lll}
\{v_{p,n}:\, p\in A_3^{(1,T-2)}\} \cup \{v_{2,n},v_{m-1,n}\} & \text{if } & m\equiv_{_5 }0,  \\ \noalign{\medskip}
	\{v_{p,n}:\, p\in A_3^{(1,T-2)}\}\cup \{v_{2,n},v_{m-2,n}\} & \text{if } & m\equiv_{_5 }1,  \\ \noalign{\medskip}
\{v_{p,n}:\, p\in A_3^{(1,T-1)}\}\cup \{v_{2,n}\}  & \text{if } & m\equiv_{_5 }2,3,4,  \\ \noalign{\medskip}
\end{array}\right.\\
D_s^{LR} &=\left\{\begin{array}{lll}
\{v_{m,p}:\, p\in A_4^{(1,S-1)}\} \cup \{v_{m,3}\} \hspace{1.1cm} & \text{if } & m\equiv_{_5 }0,  \\ \noalign{\medskip}
\{v_{m,p}:\, p\in A_2^{(0,S-1)}\}\cup \{v_{m,n-1}\} & \text{if } & m\equiv_{_5 }1,  \\ \noalign{\medskip}
\{v_{m,p}:\, p\in A_0^{(1,S)}\} , & \text{if } & m\equiv_{_5 }2,  \\ \noalign{\medskip}
\{v_{m,p}:\, p\in A_3^{(1,S)}\} \cup \{v_{m,2}\} & \text{if } & m\equiv_{_5 }3,  \\ \noalign{\medskip}
\{v_{m,p}:\, p\in A_1^{(1,S-1)}\}\cup \{v_{m,n-2}\} & \text{if } & m\equiv_{_5 }4,  \\ \noalign{\medskip}
\end{array}\right.
\end{align*}

\end{itemize}
Now, we show that the union of constructed sets builds a dominating set for $G_{m,n}$. We define $D_s$ and $D^{out}$ as follows

\begin{equation}\label{D_s}
D_s=D_s^{FR}\cup D_s^{FC}\cup D_s^{LR}\cup D_s^{LC}
\end{equation}
and
\begin{equation}
D^{out}=D_d\cup D_s.
\end{equation}

In Theorem \ref{Correct}, we will prove that the set $D^{out}$ is an optimal dominating set for $G_{m,n}$.
\begin{example}
The resulting  dominating sets for grids $G_{24,20},G_{24,21},G_{24,22},G_{24,23}$ and $G_{24,24}$  are illustrated in Figure \ref{fig:20*24}-\ref{fig:24*24}.

\end{example}

\subsection{Correctness and Time Complexity}

 We consider two case when the remainder $m$ by five be zero  or not. In the first case, we partition the grid $G_{m,n}$ into $B_1,B_2,\dots B_{T}$ such that every block $B_i$, $1\leq i\leq T$ be a grid $P_5\times P_n$.  In the second case, we divide the grid $G_{m,n}$ into $B_1,B_2,\dots B_{T+1}$ such that every block $B_i$, $1\leq i\leq T$ be a grid $P_5\times P_n$ and  the last block, $B_{T+1}$, is a grid $P_{m-5T}\times P_n$ that is denoted by $B_l$.
We remember that $T=\lfloor \frac{m}{5} \rfloor,  S=\lfloor \frac{n}{5} \rfloor$ and  the blocks are distinct.
 The number of black disks in the blocks is summarized in Table \ref{black disk}.
  
  Also the sum of all white square that locate on boundary of grid are summarized in Table \ref{white square}.

\begin{lemma}\label{domination}
	The set $D^{out}$ is a dominating set for $G_{m,n}$.
\end{lemma}
\begin{proof}
Let $v_{p,q}\in V$ and $a_p\equiv_{_5} (n +3(p-1))$, where $2\leq p\leq m-1$, be the first column in row $p$ such that a black disk is appeared. If $q= 5k+a_p$, then $v_{p,q}\in D^{out}$.

Let $p\in \{1,m\}$ and  $q\in \{1,2,n,n-1\}$. If $a_p\in \{1,2\} $, then  $v_{p,q}$ is not added to $D^{out}$ in step 1. Therefore,  there are at most four $4-$degree vertices  $\{v_{2,2},v_{2,n-1},v_{m-1,2},v_{m-1,n-1}\}$ which may be not dominated by black disks. If one of them is not dominated in step 1, then it is dominated in step 2.
		
We consider other four cases:		
\begin{itemize}		
	    \item[\bf Case 1:]$q\equiv_{_5} a_p+1$.  
		Since that all $5k'+a_p$ in row $p$ was added to $D^{out}$ then there exist a $k''$ such that $q-1=5k''+a_p$, therefor $v_{p,q-1}\in D^{out}$. 	Also $v_{p-1,q},v_{p+1,q},v_{p,q+1}$  are not selected in step 1, because $q\not\in\{a_{p-1},a_{p+1}, a_p+4\}_{mod\, 5}$.

		 \item[\bf Case 2:]$q\equiv_{_5} a_p+2$. In this case, if $p \geq 1$ then $v_{p-1,q}\in D^{out}$ because 
		\[q\equiv_{_5} a_p+2 \equiv_{_5} a_{p-1}.\]
		Also similar to previous case  $v_{p,q-1},v_{p,q+1},v_{p+1,q}$  are not selected in step 1, because $q\not\in\{a_p+1,a_p+5, a_{p+1}\}_{mod\, 5}$.
		
		\item[\bf Case 3:]$q\equiv_{_5} a_p+3$. In this case, if $p+1 \leq m$ then $v_{p+1,q}\in D^{out}$ because 
		\[q\equiv_{_5} a_p+3\equiv_{_5} a_{p+1}, \]
		and  $v_{p,q-1},v_{p,q+1},v_{p-1,q}$  are not selected in step 1, because $q\not\in\{a_p+2,a_p+4, a_{p-1}\}_{mod\, 5}$.
				
		\item[\bf Case 4:]$q\equiv_{_5} a_p+4$. Since that all $5k'+a_p$ in the row $p$ were added to $D^{out}$, then there exist a $k''$ such that $q+1\equiv_{_5} 5k''+a_p$ and $v_{p,q+1}$ was added to $D^{out}$ and  $v_{p,q-1},v_{p-1,q},v_{p+1,q}$  are not selected in step 1, because $q\not\in\{a_p+3,a_{p-1}, a_{p+1}\}_{mod\, 5}$.

		\end{itemize}
						
		It is clear that every vertex of degree four is dominated by at most one black disk. 
		
		For the boundary vertices, we discuss just the correctness of first row and the argument for the other vertices in boundary is the same manner. 
		In the first row, selection of dominating vertices only  depend on $n$ and independent from $m$. In Table \ref{first row}, we summarize black disks and white squares that are selected according to different $n$.
		\begin{table}[h]
		\caption{Dominating vertices in first row (in the case 1 and 3, $v_{1,1}$ is not dominated and in the case 4 and 5 both of $v_{1,1}$ and $v_{1,n}$ are not dominated)}\label{first row}
		\bgroup
		\def\arraystretch{1.5}
		\begin{center}
			\begin{tabular}{|c|c|c|c|}
				\hline case & n & $\text{black disks}$ & $\text{white squares}$  \\
				\hline
				 1 & $n=5l$ & $\{v_{1,p}:\, p\in A_2^{(1,S-1)}\}$ & $\{v_{1,p}:\, p\in A_4^{(1,S-1)}\}\cup\{v_{1,3}\}$   \\  		
				 \hline  
				 2 & $n=5l+1$ & $\{v_{1,p}:\, p\in A_1^{(1,S-1)}\}$ & $\{v_{1,p}:\, p\in A_3^{(1,S-2)}\}\cup\{v_{1,a},v_{1,n-2}\}$   \\   \hline
				 3 & $n=5l+2$ & $\{v_{1,p}:\, p\in A_2^{(1,S-1)}\}$ & $\{v_{1,p}:\, p\in A_4^{(1,S-2)}\}\cup\{v_{1,a+1},v_{1,n-2}\}$ 
				  \\  \hline 4 &	$n=5l+3$  & $\{v_{1,p}:\, p\in A_3^{(1,S-1)}\}$  & $\{v_{1,p}:\, p\in A_0^{(1,S-1)}\}\cup\{v_{1,n-2}\}$ \\ \hline 
				  5 & $	n=5l+4$ & $\{v_{1,p}:\, p\in A_4^{(1,S-1)}\}$ & $\{v_{1,p}:\, p\in A_1^{(1,S-1)}\}\cup\{v_{1,n-2}\}$   \\ 
				 \hline 
			\end{tabular} 			
		\end{center}
		\egroup
       
        \end{table}
         		
     For instance, let $n\equiv_{_5} 0$. Then 
     $$D^{FR}=D^{FR}_d\cup D^{FR}_s=\{v_{1,p}\,:\, p\in A_2^{1,S-1}\cup A_4^{1,S-1}\}\cup \{v_{1,3}\}.$$
     Therefore, every vertex in the first row except $v_{1,1}$, is either in the set $D^{FR}$ or it is dominated by a vertex in $D^{FR}$ . Since in all cases, the vertex $v_{2,1}$ is selected as a white square, so vertex $v_{1,1}$ also is dominated.
     
     In other cases, $n\not\equiv_{_5} 0$, all vertices in the first row are dominated except maybe $v_{1,1}$ and $v_{1,n}$. In the case 1 and 3, the vertex $v_{2,1}$ in the first column and in the case 4 and 5, the vertices  $v_{2,1}$ and $v_{2,n}$ always are selected in step 2. Hence every vertex in the first row is dominated.  
    
     \end{proof}

  \begin{center}
  	\begin{table}
  		\caption{Number of black disks in Blocks}\label{black disk}
  		\begin{tabular}{l|c|c|c|c|c|c|}
  			
  			\cline{2-7} & $m,n$ & $n=5k$ & $n=5k+1$ & $n=5k+2$ & $n=5k+3$ & $n=5k+4$ \\ 
  			\cline{2-7} First Block&$m\geq 1$ & $5S-2$ & $5S-1$ & $5S$ & $5S+2$ & $5S+3$ \\ 
  			\cline{2-7} Middle Blocks&$m\geq 1$ & $5S$ & $5S+11$ & $5S+2$ & $5S+3$ & $5S+4$ \\
  			\cline{2-7} &$m=5l$ & $5S-2$ & $5S+1$ & $5S+1$ & $5S+2$ & $5S+3$ \\ 
  			&$m=5l+1$ & $S$ & $S-1$ & $S$ & $S$ & $S$ \\ 
  			Last Block  &$m=5l+2$ & $2S-1$ & $2S$ & $2S+1$ & $2S+1$ & $2S+1$ \\ 
  			&$m=5l+3$ & $3S$ & $3S+1$ & $3S+1$ & $3S+1$ & $3S+2$ \\ 
  			&$m=5l+4$ & $4S-1$ & $4S$ & $4S$ & $4S+2$ & $4S+2$ \\ 
  			\cline{2-7} 
  		\end{tabular} 
  		
  	\end{table}
  \end{center}

  \begin{table}
  	\caption{Number of white squares in boundary}\label{white square}
  	\centering
  	\begin{tabular}{|c|c|c|c|c|c|}
  		\hline \diagbox[dir=NW]{n}{m} & $5k$ & $5k+1$ & $5k+2$ & $5k+3$ & $5k+4$  \\  \hline 
  		
  		$n=5l$ & $2T+2S$ & $2T+2S-1$ & $2T+2S$ & $2T+2S-1$ & $2T+2S+2$ \\
  		$n=5l+1$ & $2T+2S-1$ & $2T+2S+1$ & $2T+2S-1$ & $2T+2S-1$  &$2T+2S+1$ \\ 
  		$n=5l+2$ &  $2T+2S+1$ & $2T+2S $& $2T+2S$ & $2T+2S$ & $2T+2S+2$ \\
  		$n=5l+3$  &  $2T+2S-1$ & $2T+2S$ & $2T+2S$ & $2T+2S+1$ &$2T+2S+1$ \\ 
  		$n=5l+4$  & $2T+2S+1$ & $2T+2S+1$ & $2T+2S$ & $2T+2S+1$ & $2T+2S$ \\
  		\hline 
  	\end{tabular}

  \end{table}

      	So,	every vertex of degree four is dominated by at most one black disk.

      \begin{lemma}\label{number}
      Let $D^{out}$ be the output of our method, then
      \[|D^{out}|=\left \lfloor\frac{(n+2)(m+2)}{5}\right \rfloor-4.\]
      \end{lemma}
      \begin{proof}
      It is straightforward to investigate  $|D^{out}|=\gamma (G_{m,n})$ as follows.
     
     The black disks of the set $D^{out}$ are divided into three part: black disks in the first block, middle blocks and last block.
	Hence 	
		\[D_d =\left\{ \begin{tabular}{ l l} 
						$ D_d(B_1)+(T-2)D_d(B_i)+D_d(B_l)$, &   \text{if }$m=5k$, \\  \noalign{\medskip}
						$D_d(B_1)+(T-1)D_d(B_i)+D_d(B_l)$, &  \text{otherwise}. 
						\end{tabular} \right. \]	
						
For instance, we show for $m,n$ if they are multiples of five, $D^{out}=D_d\cup D_s$, where $D_s$ is defined in Eq.\eqref{D_s}. The other cases are similar.
		\[D^{out}=D_d\cup D_s=5S-2+(T-2)(5S)+ 5S-2+2T+2S=\left \lfloor\frac{(n+2)(m+2)}{5}\right \rfloor-4.\]
      \end{proof}
      
      \begin{lemma}\label{time}
      	The set $D^{out}$ can be computed in time $O(\text{size of answer})$.
      	\end{lemma}
      	
      	\begin{proof}
      		By lemma \ref{number}, $|D^{out}|$ is equal to $\gamma(G_{m,n})$.
      	\end{proof}
      
     According to Lemmas \ref{domination}, \ref{number} and \ref{time}, we have the following theorem.
\begin{theorem}\label{Correct}
	The set  $D^{out}$ is an optimal dominating set for $G_{m,n}$ and is computed in $O(\text{size of answer})$ time.
	\end{theorem}

	\section{$[1,2]$-Domination number of Grid}
	
	In this section, we follow the construction is proposed in the Section \ref{construct} to obtain a $[1, 2]$-dominating set for grid $G_{m,n}$, where $16 \leq m \leq n$. It is not hard to investigate that $\gamma_{[1,2]}(G_{m,n})$ for  $n,m\leq 16$ by constructions are presented in \cite{chang1994domination}.
	We also show that   $\gamma_{[1,2]}(G_{m,n})=\gamma(G_{m,n})$ where $m,n\geq 16$ which is a positive answer to open question is posed in \cite{chellali20131}.

	\begin{theorem}
	Let $m,n\geq 16 $ and $G_{m,n}$ be a $m\times n$ grid, then
	\[\gamma_{[1,2]}(G_{m,n})=\gamma(G_{m,n}).\]
	\end{theorem}
	\begin{proof}	
	We show that every vertex $v_{p,q}\in V$ is dominated by at most two vertices in $D^{out}$ and according to the $\gamma_{[1,2]}(G)\geq \gamma(G)$, the result is obtained.
	
	In proof of Lemma \ref{number}, we show that every vertex of sub-grid is dominated exactly by one black disk. Also,  the distance between every two white squares is at least $5$. Then every vertex $v_{p,q}$ is dominated at most twice.

     In Table \ref{Last},  the white squares in  $D^{out}$ that dominate vertices $v_{2,2},v_{2,n-1},v_{m-1,2}$ and $v_{m-1,n-1}$ are shown.

  \begin{table}[h] 
   \caption{White squares that dominate $\{v_{2,2},v_{2,n-1},v_{n-1,2},v_{n-1,n-1}\}$}
  \centering
  \resizebox{\textwidth}{!}{
    			\begin{tabular}{|c|c|c|c|c|c|}
    				\hline  \diagbox[dir=NW]{n}{m} & $5k$ & $5k+1$ & $5k+2$ & $5k+3$ & $5k+4$  \\  \hline 
    			
    				 $5l$ & $\{v_{2,1},v_{1,n-1},v_{n,2},v_{n-1,n}\}$ & $\{v_{2,1},v_{1,n-1},v_{n-1,n}\}$ & $\{v_{2,1},v_{1,n-1},v_{n,2},v_{n,n-1}\}$ & $\{v_{2,1},v_{1,n-1},v_{n-1,1}\}$ & $\{v_{2,1},v_{1,n-1},v_{n-1,1},v_{n,n-1}\}$ \\  \hline 
    				 $5l+1$ & $\{v_{1,2},v_{n-1,n}\}$ & $\{v_{1,2},v_{n,n-1}\}$ & $\{v_{1,2},v_{n,2}\}$  & $\{v_{1,2},v_{n,n-1}\}$ & $\{v_{1,2},v_{n-1,n}\}$ \\  \hline
    				 $5l+2$ &  $\{v_{2,1},v_{n,2},v_{n-1,n}\}$ & $\{v_{2,1}\}$ & $\{v_{2,1},v_{n,2}\}$ &$\{v_{2,1},v_{n-1,1},v_{n,n-1}\}$ &$\{v_{2,1},v_{n-1,1},v_{n-1,n}\}$ 
    				  \\  \hline
    				  	$5l+3$  &  $\{v_{n,2},v_{n-1,n}\}$  &  $\{v_{n-1,1},v_{n,n-1}\}$ &  $\{v_{n-1,1}\}$ &  $\{v_{n,2},v_{n,n-1}\}$ &  $\{v_{n-1,n}\}$ \\  \hline 
    				  	$5l+4$  & $\{v_{1,2},v_{n-1,1},v_{n-1,n}\}$  & $\{v_{1,2},v_{n,2},v_{n,n-1}\}$  & $\{v_{1,2}\}$ & $\{v_{1,2},v_{n,2},v_{n,n-1}\}$ & $\{v_{1,2},v_{n-1,1},v_{n-1,n}\}$\\ 
    				  \hline 
    			\end{tabular}} 
    	 
   \end{table}\label{Last}

    By Table \ref{Last}, it can be seen that none of white square pairs $\{v_{1,2},v_{2,1}\}$, $\{v_{2,n},v_{1,n-1}\}$, $\{v_{n-1,1},v_{n,2}\}$ and $\{v_{n-1,n},v_{n,n-1}\}$
     appears in any cell. So the vertices  $\{v_{2,2},v_{2,n-1},v_{n-1,2},v_{n-1,n-1}\}$ are not dominated more than two times.
    
    These claims are appeared in Figures \ref{fig:LU} and \ref{fig:RU}. In fact, Figures \ref{fig:LU-a1}-\ref{fig:LU-a4}, show that the vertex $v_{2,2}$ is dominated at most twice according to different $a_1$. Since selecting the domination vertecis at the right-up corner depend on $a_1$ and $n$. Therefor two cases occur, as can be seen in Figures \ref{fig:RU-n0} and \ref{fig:RU-n1234}. These figures show that the vertex $v_{2,n-1}$ is dominated at most twice.  For the other corners, we have a similar argument. 
    
    \begin{figure}[h!]    
            \centering
            \begin{subfigure}[h]{0.23\textwidth}
                \centering
                \includegraphics[width=\textwidth]{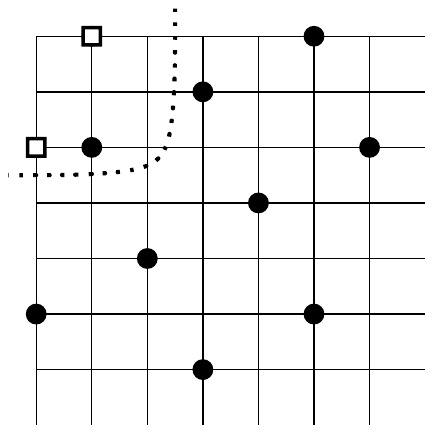}
                \caption{$a_1=1$}
                \label{fig:LU-a1}
            \end{subfigure}
            \hfill
            \begin{subfigure}[h]{0.23\textwidth}
                \centering
                \includegraphics[width=\textwidth]{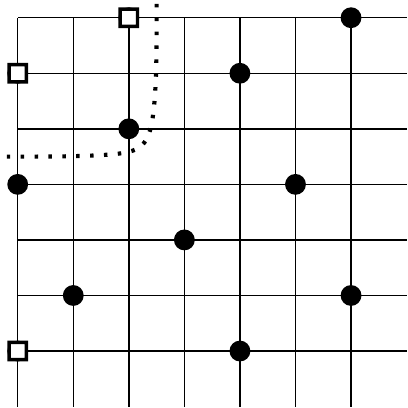}
                \caption{$a_1=2$}
                \label{fig:LU-a2}
            \end{subfigure}
            \hfill
            \begin{subfigure}[h]{0.23\textwidth}
                   \centering
                        \includegraphics[width=\textwidth]{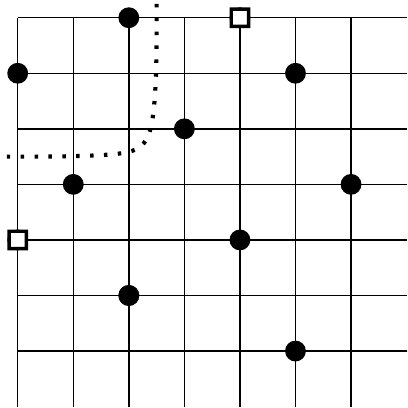}
                        \caption{$a_1=3$}
                        \label{fig:LU-a3}
                    \end{subfigure}
                    \hfill
                    \begin{subfigure}[h]{0.23\textwidth}
                        \centering
                        \includegraphics[width=\textwidth]{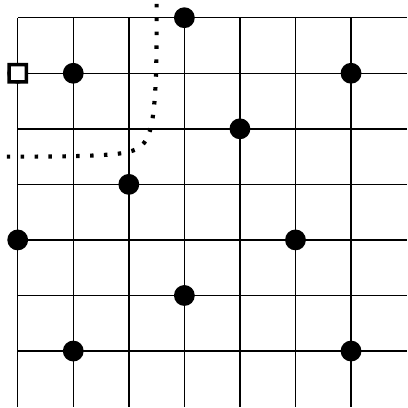}
                        \caption{$a_1=4$}
                        \label{fig:LU-a4}
                    \end{subfigure}
          
            \caption{Left-Up corner of grid}\label{fig:LU}
        \end{figure}
        
        \begin{figure}[h!]
            \centering
             \begin{subfigure}[h]{0.45\textwidth}
               \centering
               \includegraphics[width=.55\textwidth]{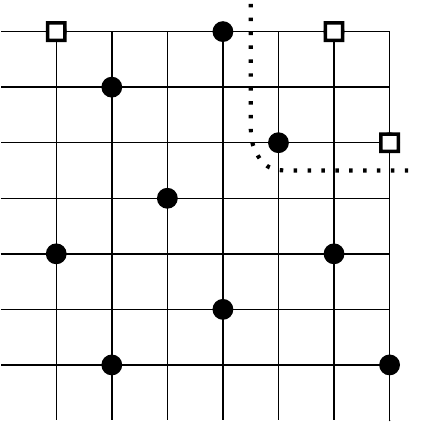}
                \caption{$n\equiv_{_5} 0$}
                \label{fig:RU-n0}
            \end{subfigure}
        ~
            \begin{subfigure}[h]{0.45\textwidth}
                \centering
                \includegraphics[width=.55\textwidth]{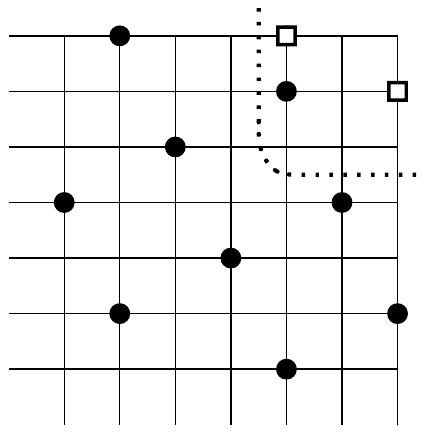}
                \caption{$n\not\equiv_{_5} 0$}
                \label{fig:RU-n1234}
            \end{subfigure}
            \hfill
            \caption{Right-Up corner of grid}\label{fig:RU}
        \end{figure}
        So, every vertex of $G_{m,n}$ is dominated at least one and at most twice by vertices of $D^{out}$.
	\end{proof}
	
	\begin{figure}[h]
		\centering
		\begin{subfigure}[h]{0.45\textwidth}
			\centering
			\includegraphics[scale=0.5]{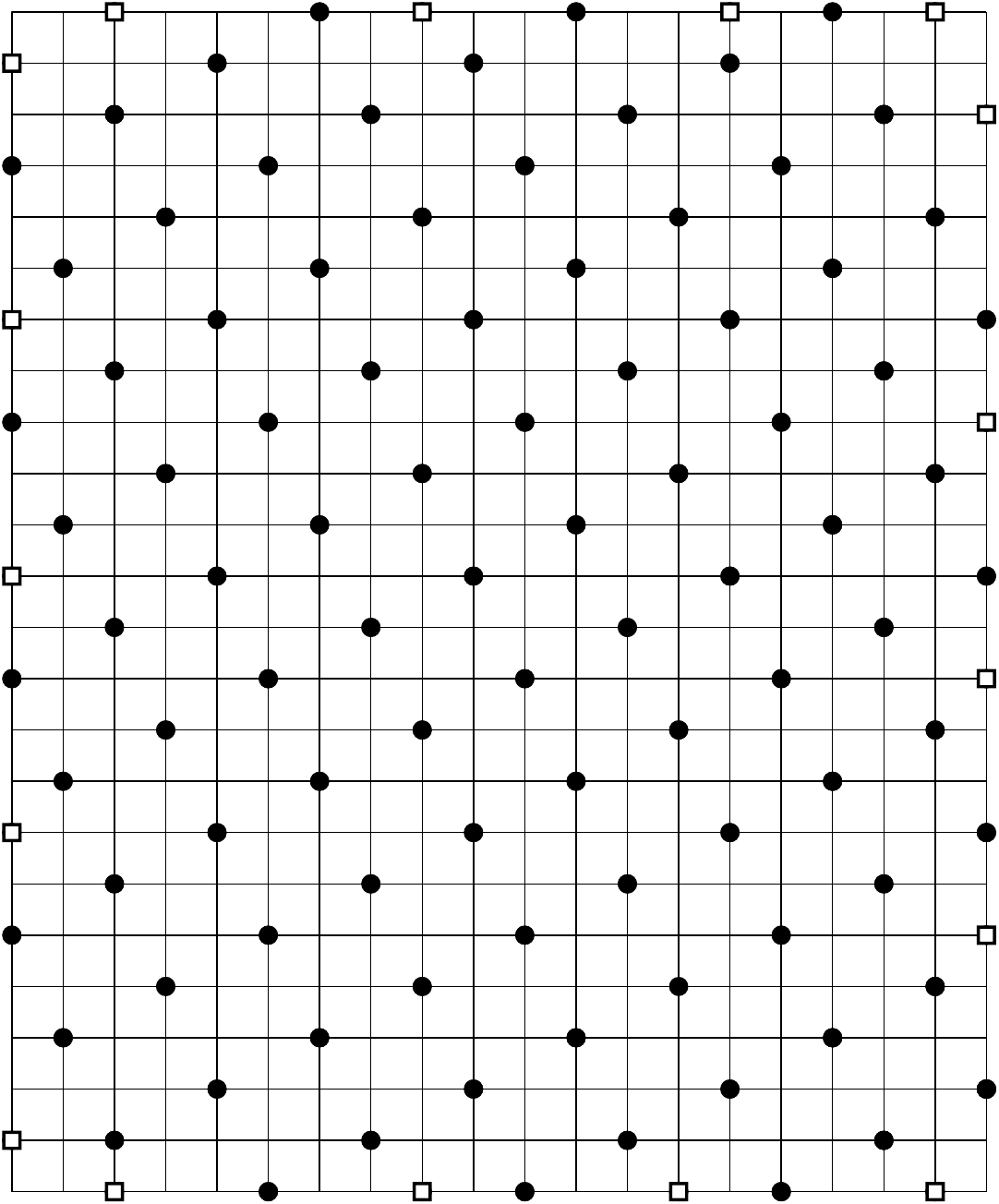}
			\caption{Dominationg set of size 110 in $G_{24,20}$}
			\label{fig:20*24}
		\end{subfigure}
		~
		\begin{subfigure}[h]{0.45\textwidth}
			\centering
			\includegraphics[scale=0.5]{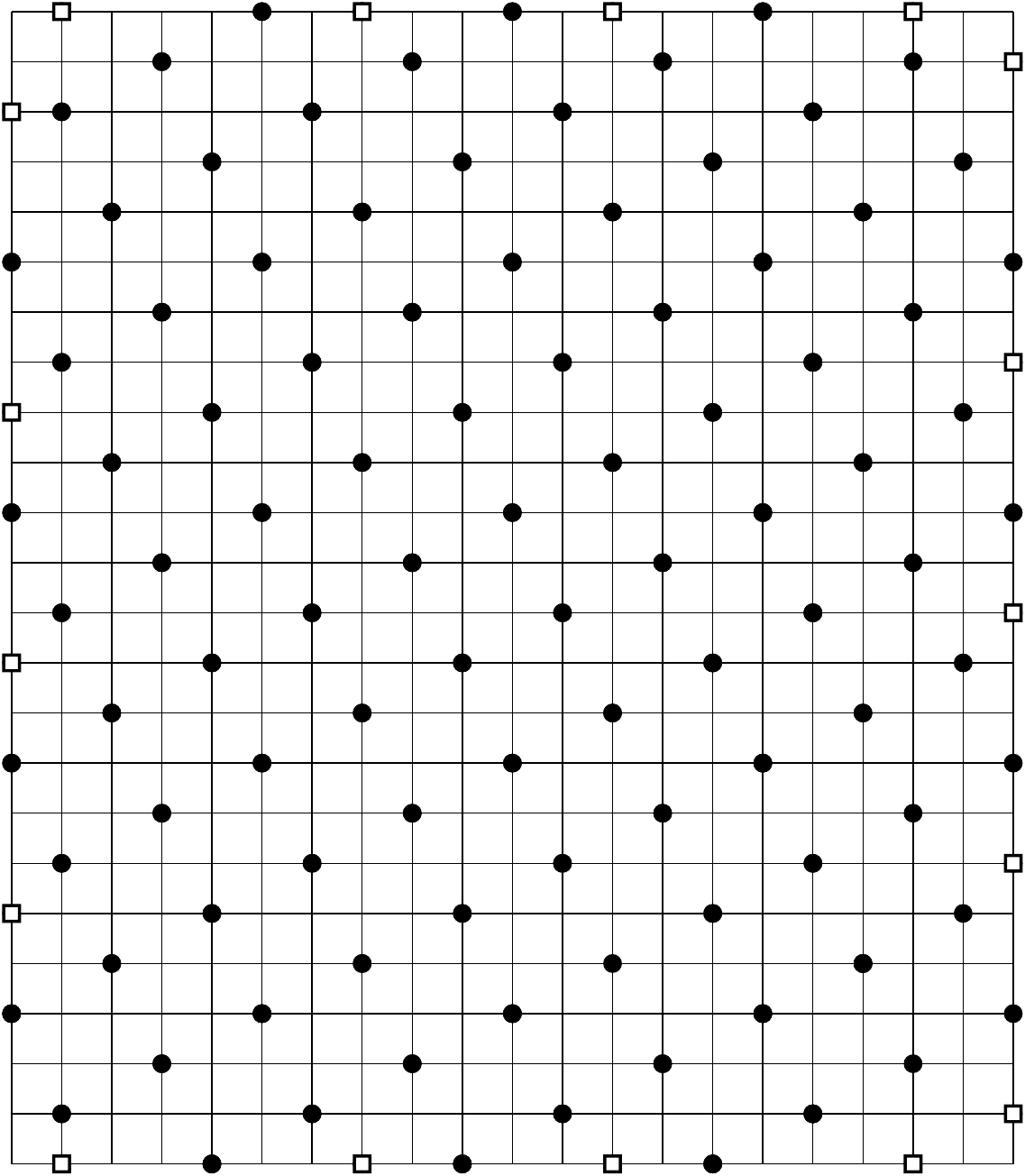}
			\caption{Dominationg set of size 115 in $G_{24,21}$}
			\label{fig:21*24}
		\end{subfigure}
		
		\centering
		\begin{subfigure}[h]{0.45\textwidth}
			\centering
			\includegraphics[scale=0.5]{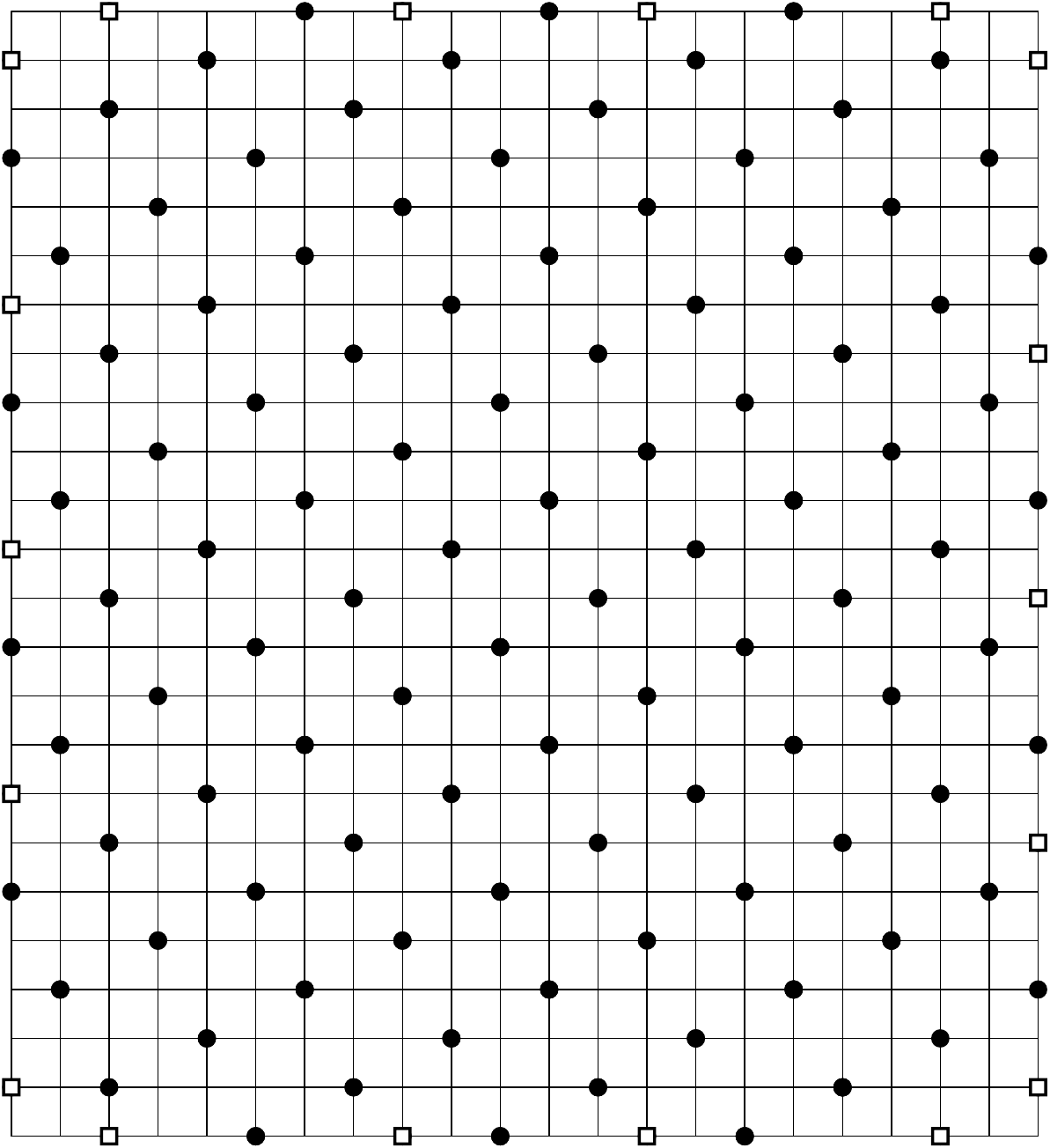}
			\caption{Dominationg set of size 120 in $G_{24,22}$}
			\label{fig:22*24}
		\end{subfigure}
		~
		\begin{subfigure}[h]{0.45\textwidth}
			\centering
			\includegraphics[scale=0.5]{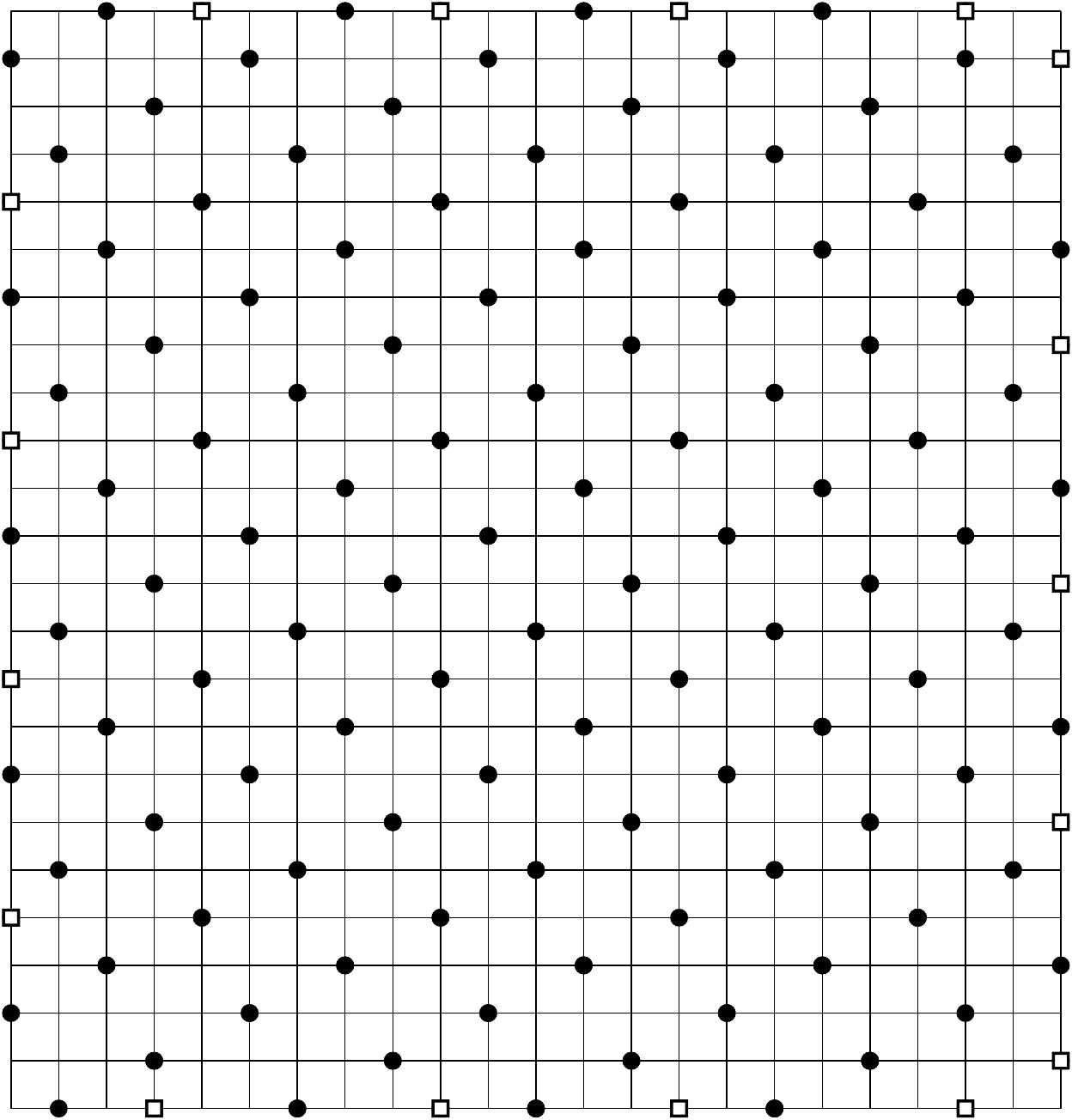}
			\caption{Dominationg set of size 126 in $G_{24,23}$}
			\label{fig:23*24}
		\end{subfigure}
		\hfill
		\begin{subfigure}[h]{0.45\textwidth}
			\centering
			\includegraphics[scale=0.5]{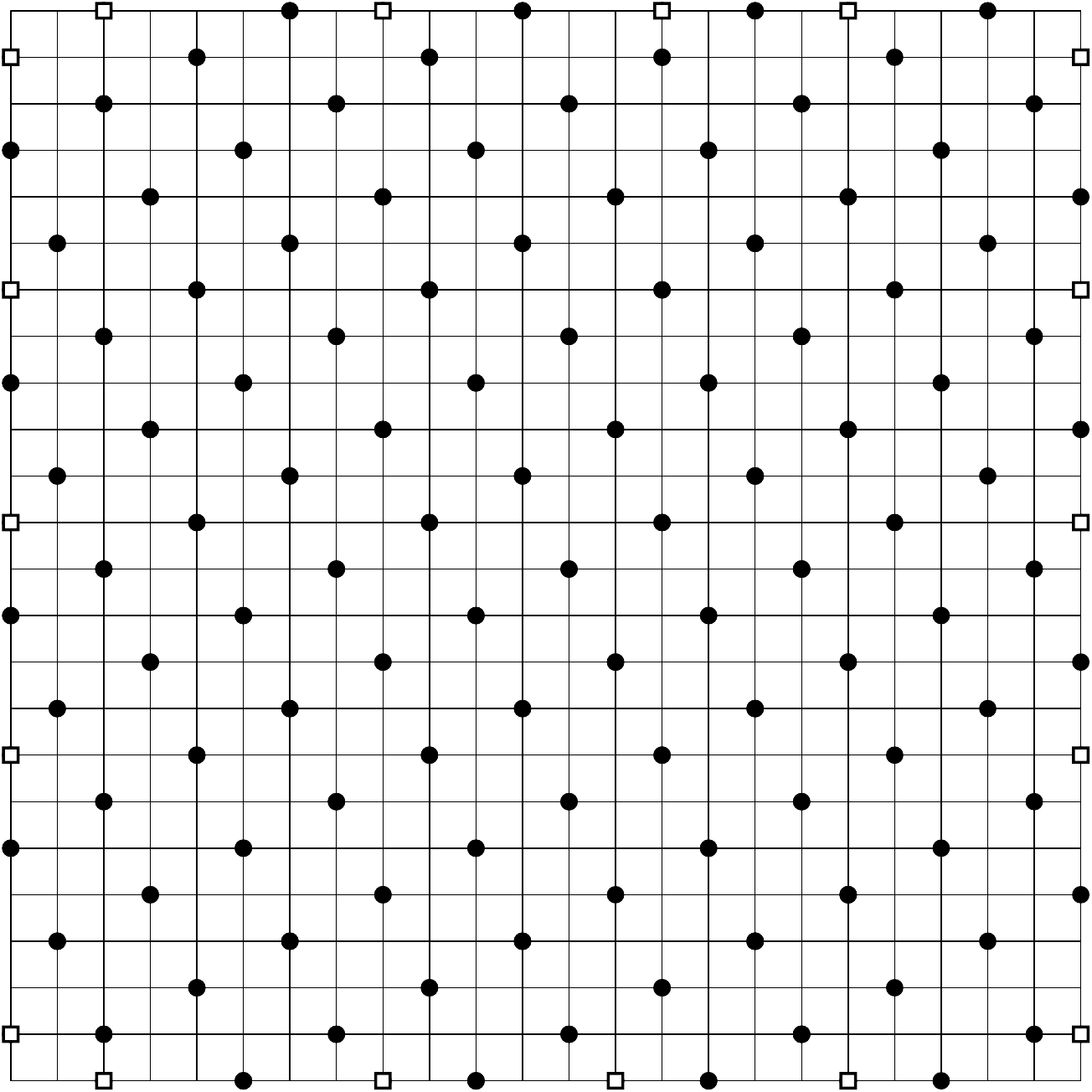}
			\caption{Dominating set of size 131 for grid $G_{24,24}$}
			\label{fig:24*24}
		\end{subfigure}
		
		\caption{Examples of some grids}
	\end{figure}


	\end{document}